\documentclass[conference]{IEEEtran}
\IEEEoverridecommandlockouts

\usepackage{cite}
\usepackage{amsmath,amssymb,amsfonts,amsthm}
\usepackage{algorithmic}
\usepackage{graphicx}
\usepackage{textcomp} 
\usepackage{hyperref}
\usepackage{tabularx}
\usepackage{booktabs}
\usepackage{xcolor}
\newtheorem{proposition}{Proposition}
\newtheorem{definition}{Definition}
\def\BibTeX{{\rm B\kern-.05em{\sc i\kern-.025em b}\kern-.08em
    T\kern-.1667em\lower.7ex\hbox{E}\kern-.125emX}}
\begin{document}

\title{How Much Does Submitting a Burst of Jobs Impact Waiting Time? An Analysis of the Jean Zay Workload Trace\\

    \thanks{We would like to thank the IDRIS team for providing the workload traces}
}
\author{\IEEEauthorblockN{1\textsuperscript{st} Marius Garénaux}
    \IEEEauthorblockA{\textit{IRISA, Université de Rennes} \\
        Rennes, France \\
        \url{https://orcid.org/0009-0003-7178-2812}}
    \and
    \IEEEauthorblockN{2\textsuperscript{nd} François Bodin}
    \IEEEauthorblockA{\textit{IRISA, Université de Rennes} \\
        Rennes, France \\
        \url{https://orcid.org/0000-0002-0928-2564}}
}
\maketitle

\begin{abstract}
    This paper examines job waiting times on the Jean Zay supercomputer at
    IDRIS (\textit{Institut du Développement et des Ressources en Informatique
        Scientifique \url{http://www.idris.fr/}}), driven by the need to improve
    scheduling strategies for cross-facility scientific workflows
    \cite{ref-ferreiradasilvaFrontiersScientificWorkflows2024}, \cite{ref-antypasEnablingDiscoveryData2021},
    \cite{ref-deelmanFutureScientificWorkflows2018}. Drawing on
    large workload traces, we show that job waiting times are
    not independent of one another but are instead correlated through
    job bursts—batches of identical jobs submitted together, also called ``flurries''
    \cite{ref-feitelsonFlurries}. We introduce a
    method for quantifying the effect of these bursts on waiting times, and find
    a strong Pearson correlation between waiting time and a job's rank within
    its burst. This correlation is above 0.8 for more than 40\% of the jobs.
    We further derive an exact formula for the Slurm fairshare factor based
    on a convolution product, showing that the burst
    penalty operates on a much shorter timescale than fairshare itself—pointing to a
    distinct throttling mechanism specifically targeting bulk submissions. Once
    burst effects are accounted for, we use Welch's t-test to find
    that partition occupancy and number of requested cores have a significant effect
    on waiting time, though the effect is not significant for all partitions.
\end{abstract}


\section{Introduction}\label{introduction}

In the evolution of scientific workflows from single-facility to
distributed (cross-facility) deployments
\cite{ref-ferreiradasilvaFrontiersScientificWorkflows2024},
\cite{ref-antypasEnablingDiscoveryData2021},
\cite{ref-deelmanFutureScientificWorkflows2018}, managing
uncertainty in resource allocation has become an increasingly important concern. This
shift is driven by several factors, including machine switching during maintenance,
the ability to match specific tasks to specialized hardware architectures, the
potential to reduce the overall makespan through high-level parallelism across
multiple sites, and the need to process massive data volumes close to their
source of production \cite{ref-antypasEnablingDiscoveryData2021}. Typically,
executing code on a supercomputer involves requesting a specific number
of cores for a given duration. A batch scheduler (such as Slurm, PBS,
or OAR) then queues the job until the requested resources become
available. While cross-facility workflows---where tasks are distributed
across multiple supercomputers---face these same fundamental limitations, the
systemic impact is amplified. A delay in a single task can propagate
through the entire workflow, adversely affecting the overall makespan.
Consequently, a precise characterization of
waiting times is essential to inform the design of future high-level
workflow orchestration methods.

Predicting this latency is notoriously difficult due to the complex
interplay between batch scheduler heuristics, hardware constraints, and
user behavior. To
investigate these dynamics, this study utilizes two workload traces
from the Jean Zay supercomputer (hosted at IDRIS). These traces
comprise approximately 80,000 and 95,000 pseudonymized jobs,
spanning respectively one week in October 2025 and one week in March 2026. The former
was used for exploratory data analysis, while the
latter was reserved exclusively for validation.

In Section \ref{related-works}, we review previous work on waiting time analysis and prediction.
In Section \ref{dataset-presentation-jean-zay-workload-traces}, we introduce
the Jean Zay workload traces, highlighting the strong disparities between users.
We then examine the impact of job bursts on waiting times in Section
\ref{the-phenomenon-of-job-bursts}, showing that more than 40\% of
jobs' waiting times are determined by their rank within a burst. In Section
\ref{slurm-fairshare-factor}, we derive an exact formulation of the Slurm
fairshare factor based on a convolution product. This shows that the observed
``burst penalty'' operates on a much shorter timescale than the fairshare factor
and is therefore unlikely to originate from the fairshare mechanism itself. Finally,
in Section \ref{analysis-of-waiting-times}, we exclude jobs belonging to bursts
and analyze the remaining jobs. We show that jobs submitted during periods of
lower partition occupancy also experience shorter waiting times. We also find
that jobs that requested fewer cores tend to have lower waiting times.

\section{Related Works}\label{related-works}

The prediction of queue waiting times has been recognized as a
challenging task in the HPC community due to high variability and highly
skewed distributions
\cite{ref-feitelsonWorkloadModelingComputer2015},
\cite{ref-okaforQueueWaitTime2024}. This difficulty stems
from the multi-dimensional nature of the problem, where latency is
influenced by supercomputer policies, the internal state of the batch
scheduler, current queue congestion, and user-specific priority factors
like fairshare.

Early efforts in the late 1990s and early 2000s focused on the analysis
of workload traces to identify patterns and establish baseline
predictions \cite{ref-downeyPredictingQueueTimes1997},
\cite{ref-smithPredictingApplicationRun1998}. While some
research has shifted toward grid computing, HPC-specific studies
continue to evolve, often driven by diverse motivations. Some works aim
to improve user experience and resource sharing
\cite{ref-iosupHowAreReal2006},
\cite{ref-parkQueueCongestionPrediction2022}, while others
focus on high-level workflow orchestration
\cite{ref-nurmiEvaluationWorkflowScheduler2006},
\cite{ref-vercellinoMachineLearningApproach2023},
\cite{ref-stubbsSmartSchedulingTapis2024} or the specific
constraints of urgent computing
\cite{ref-brownPredictingBatchQueue2022}. Methodologically, predictions have evolved from similarity-based methods
such as k-Nearest Neighbors (k-NN)
\cite{ref-kianpishehPredictingJobWait2012},
\cite{ref-liEfficientResponseTime2005}, simulation and time series prediction \cite{ref-sonmezTracebasedEvaluationJob2009},
to more complex probabilistic models using Hidden Markov model
\cite{ref-parkQueueCongestionPrediction2022} and, more
recently, deep learning architectures
\cite{ref-lovellHierarchicalDeepLearning2024}. The Jean Zay
trace used in this paper is not publicly available. Despite
user pseudonymization, specific behaviors creates highly identifiying
patterns. Availability of such high quality traces is fundamental
for this kind of work, and yet they are not easily accessible. We
would like to thank again the IDRIS team for providing this workload
trace.

The submission of a sequence of identical jobs (referred to as ``job bursts'', or ``flurries'') has been identified as a known behavior in workload
traces \cite{ref-lovellHierarchicalDeepLearning2024} \cite{ref-feitelsonWorkloadModelingComputer2015} \cite{ref-feitelsonFlurries}.
Here, we go further by measuring the proportion of jobs affected
and quantifying the resulting impact on waiting times.

\section{Jean Zay Workload
  Traces}\label{dataset-presentation-jean-zay-workload-traces}

The two workload traces were sourced by the IDRIS team from
Slurm \cite{ref-ArchitectureSlurmWorkload}---which is the
batch scheduler used on Jean Zay---and from CEEMS
\cite{ref-10.1109ux2fSCW63240.2024.00233}.

\subsection{Jean Zay}

Jean Zay is one of the three French Tier-0 supercomputers. It consists of five
partitions\footnote{http://www.idris.fr/docs/jean-zay/jean-zay/jean-zay-hardware}: one scalar
partition (CPU only) and four accelerated partitions (CPU+GPU). The system is operated by IDRIS.

The cores listed in Table \ref{tab:jean-zay-partitions} are configured with hyperthreading, with two threads per core.
The number of cores per CPU varies across partitions. For instance,
each CPU in the cpu\_p1 partition has 20 physical cores, and
each node contains two CPUs, for a total of 40 physical cores per node. Note
that Slurm refers to hardware threads as CPUs ---which can be misleading. Throughout
the remainder of this work, the term core always refers to a physical
CPU core rather than a hardware thread, and the term CPU \textbf{never} refers
to either cores or threads (except in the Table \ref{tab:dataset-fields},
describing Slurm logs columns).

\begin{table}[ht]
    \centering
    \caption{Characteristics of Jean Zay partitions.} \label{tab:jean-zay-partitions}
    \begin{tabular}{llccccc}
        \toprule Partition & Name             & n nodes & core/node & GPUs/node \\
        \midrule
        CPU                & \texttt{cpu\_p1} & 720     & 40        & --        \\
        Quad-GPU V100      & \texttt{p13}     & 396     & 40        & 4         \\
        Octa-GPU V100      & \texttt{p2}      & 31      & 24        & 8         \\
        Octa-GPU A100      & \texttt{p5}      & 52      & 64        & 8         \\
        Quad-GPU H100      & \texttt{p6}      & 364     & 96        & 4         \\
        \bottomrule\end{tabular}

\end{table}

Each partition provides several Quality of Service (QoS)\footnote{http://www.idris.fr/docs/jean-zay/slurm/slurm\_partitions\_cpu} classes,
which define limits on job duration and resource allocation. For example,
on the \texttt{cpu\_p1} partition, jobs exceeding 20 hours cannot use the
\texttt{t3} QoS and must instead request the \texttt{t4} QoS.

\subsection{Columns and Pre-processing}\label{columns-and-pre-processing}

The extracted columns are described in Table \ref{tab:dataset-fields}. Additional
columns were computed, as described in Table \ref{tab:dataset-additional-fields}.

\begin{table}[ht]
    \centering
    \caption{Description of the features}
    \label{tab:dataset-fields}
    \begin{tabularx}{.47\textwidth}{@{}lX@{}}
        \toprule
        \textbf{Field} & \textbf{Description}                                                                                                                                                                                \\
        \midrule
        JobID          & Identification number of the job or job step.                                                                                                                                                       \\

        Account        & Pseudonymized account identifier.                                                                                                                                                                   \\

        User           & Pseudonymized user identifier.                                                                                                                                                                      \\

        State          & The state of the job.                                                                                                                                                                               \\




        Create Time    & Submission time of the job.                                                                                                                                                                         \\

        Elapsed        & The job's elapsed time, in seconds.                                                                                                                                                                 \\

        End Time       & Termination time of the job.                                                                                                                                                                        \\

        Num CPUs       & Total number of threads allocated to the job. For Jean Zay, this is \href{https://slurm.schedmd.com/faq.html\#cpu\_count}{twice the number of
        physical cores}                                                                                                                                                                                                      \\

        Num GPUs       & Total number of GPUs allocated to the job.                                                                                                                                                          \\

        Partition      & Identifies the partition on which the job ran.                                                                                                                                                      \\

        QoS            & Name of the Quality of Service for the job. \href{http://www.idris.fr/docs/jean-zay/slurm/slurm\_partitions\_cpu}{CPU}, \href{http://www.idris.fr/docs/jean-zay/slurm/slurm\_partitions\_gpu}{GPU}. \\

        Start Time     & Initiation time of the job.                                                                                                                                                                         \\

        \bottomrule
    \end{tabularx}
\end{table}

\begin{table}[ht]
    \centering
    \caption{Description of additional features}
    \label{tab:dataset-additional-fields}
    \begin{tabularx}{.47\textwidth}{@{}lX@{}}
        \toprule
        \textbf{Field}          & \textbf{Description}                                                                                                         \\
        \midrule                                                                                                                                               \\
        waiting\_time           & Calculated time spent in queue                                                                                               \\
        effective\_cores        & Actual number of cores booked, since the batch scheduler
        can change the number of requested resources to make it match a full
        node or a fixed portion of a node. See \href{http://www.idris.fr/docs/jean-zay/slurm/slurm\_partitions\_cpu}{Jean Zay documentation}.                  \\
        effective\_gpus         & Actual GPU resources used, since the batch scheduler
        can change the number of requested resources to make it match a full
        node or a fixed portion of a node. See \href{http://www.idris.fr/docs/jean-zay/slurm/slurm\_partitions\_gpu}{Jean Zay documentation}                   \\
        cores\_seconds          & Total cores$\times$seconds used for this job (number of cores times duration of job)                                         \\
        gpus\_seconds           & Total GPU$\times$seconds used for this job (number of GPUs times duration of job)                                            \\
        user\_historical\_usage & Historical resource usage for the user. See \href{https://slurm.schedmd.com/classic\_fair\_share.html}{Slurm documentation}. \\
        partition\_core\_prop   & Proportion of cores used in the partition,
        at the moment of job submission                                                                                                                        \\
        \bottomrule
    \end{tabularx}
\end{table}

Among all jobs, we dropped
``REQUEUED'' and ``NODE\_FAIL'' jobs to remove duplicate Job IDs. This resulted
in a trace of around 65,000 jobs for the week in October 2025 and 80,000 for the week in March 2026. Metrics
related to the Slurm fairshare algorithm were recomputed using several
sources to obtain the exact computation method. Namely, we used the Slurm source code \footnote{\url{https://github.com/SchedMD/slurm/blob/f0884a3e/src/plugins/priority/multifactor/priority\_multifactor.c\#\#L189-L202}},
documentation \footnote{\url{https://slurm.schedmd.com/classic\_fair\_share.html}}
and DeepWiki \footnote{\url{https://deepwiki.com/}} to understand the code.

\subsection{Partition and QoS}\label{partition-and-qos-usage}

Preliminary analysis of the October trace (Figure \ref{hist_usage}) indicates that the
\texttt{gpu\_p13} and \texttt{cpu\_p1} partitions account for 73\% of total number of job submissions,
with a predominance of the \texttt{t3} Quality of Service (QoS). However, this metric does
not take into account the duration and the resources asked for each job, which are summarized in
Figure \ref{avg_core_second}.

\begin{figure}
    \centering
    \includegraphics[width=\linewidth]{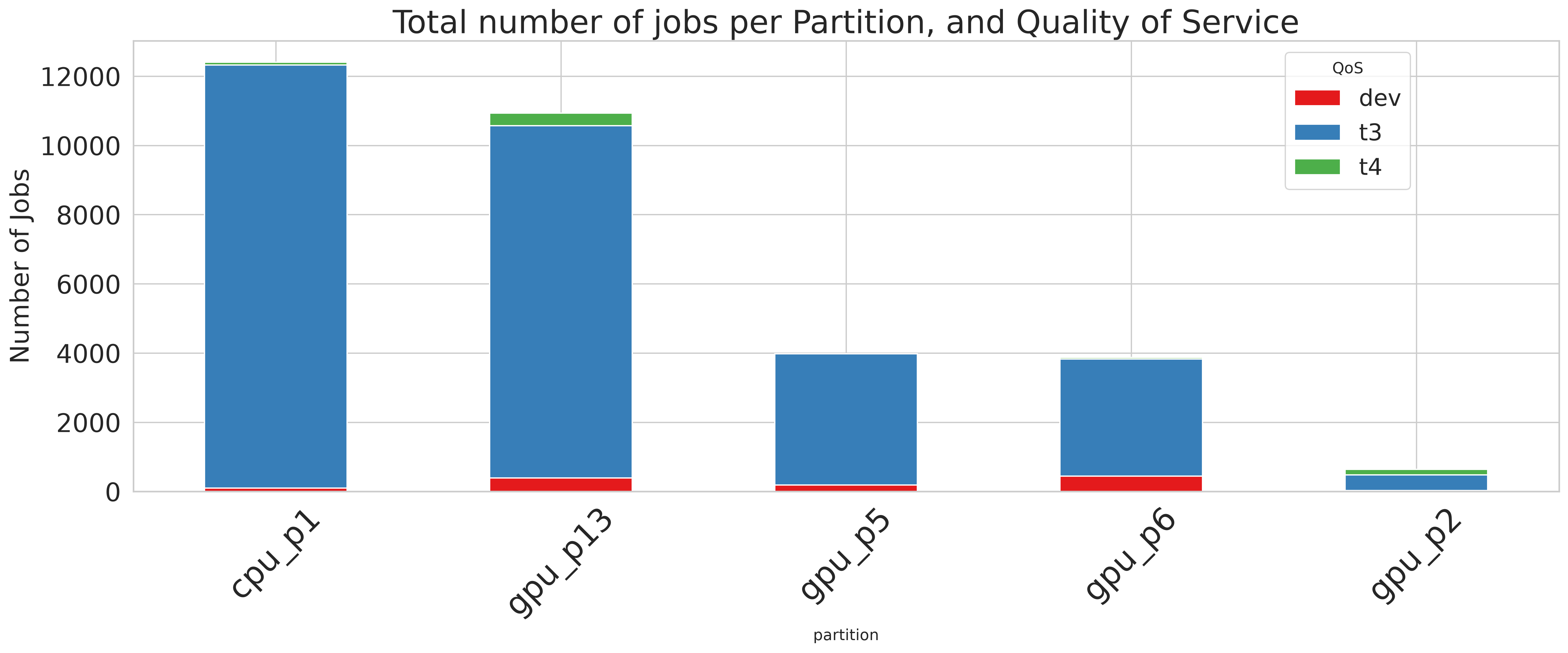}
    \caption{Usage of partition, with different Quality of Service - Jean
        Zay (IDRIS) - 15/10/25 21/10/25 \label{hist_usage}}
\end{figure}

The Jean Zay supercomputer operates under sustained demand, as
illustrated in Figures \ref{jz_charge}, \ref{avg_core_second}. Such high demand is
characteristic of this supercomputer and may have been exacerbated by the
end of the resource attribution period in October, and April.
The H100 (\texttt{p6}) partition, in particular, maintains a mean occupancy rate exceeding
90\%. For all analysis requiring partition occupancy, the first and last days of the traces were excluded to
remove partial records---specifically, jobs that started before the
collection window or remained running after it closed.

\begin{figure}
    \centering
    \includegraphics[width=\linewidth]{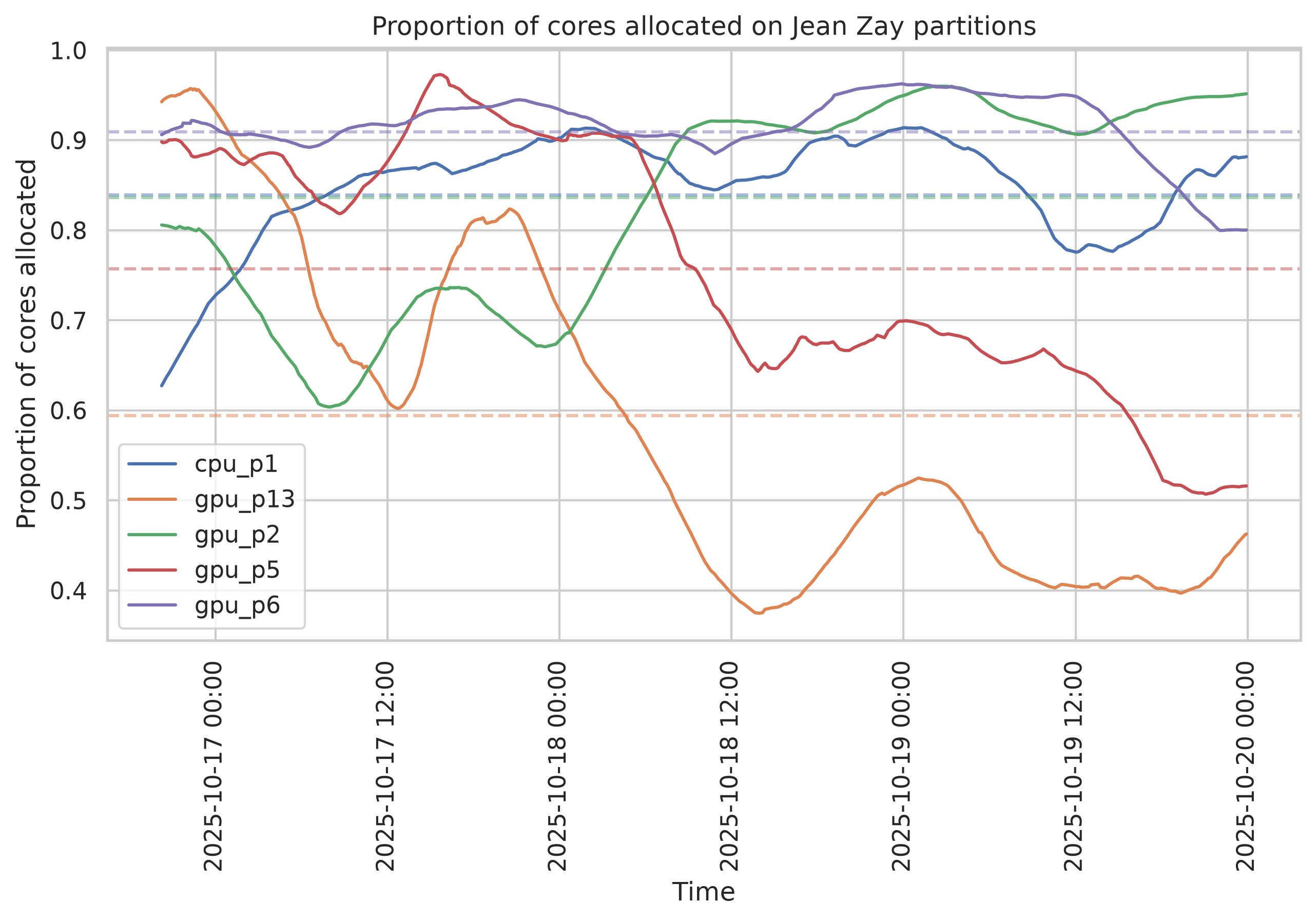}
    \caption{Occupancy of each partition (the dotted lines represent the average) -- Jean Zay (IDRIS)\label{jz_charge}}\end{figure}

\begin{figure}
    \centering
    \includegraphics[width=\linewidth]{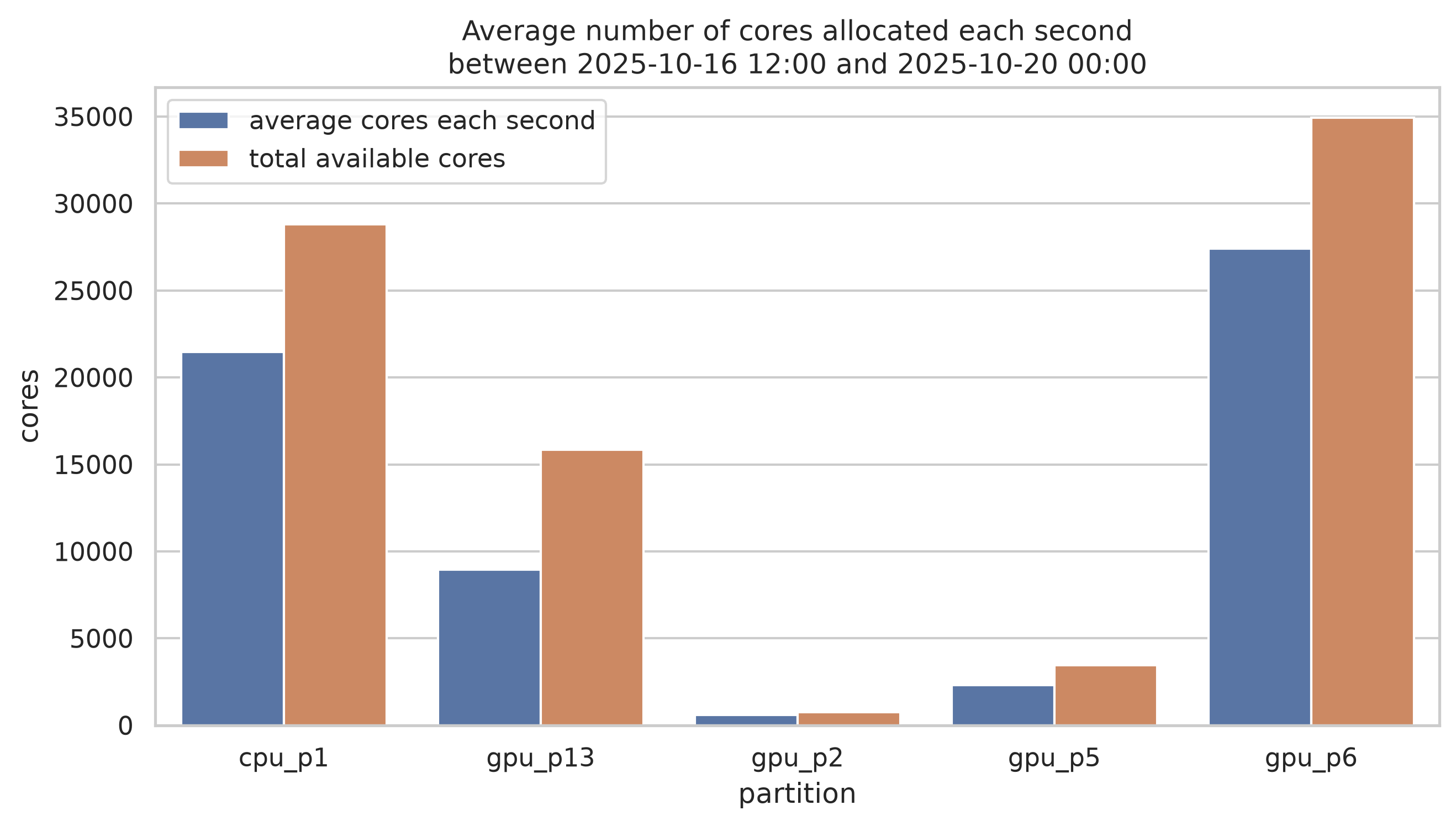}
    \caption{Average number of cores used each second, among partitions \label{avg_core_second}}
\end{figure}

Finally, the volume of job submissions is significantly lower on
weekends, as shown in Figure \ref{job_over_time}. This trend correlates
with the lower occupancy rates observed for the \texttt{gpu\_p5} and
\texttt{gpu\_p13} partitions during these periods (Figure
\ref{jz_charge}).

\begin{figure}
    \centering
    \includegraphics[width=\linewidth]{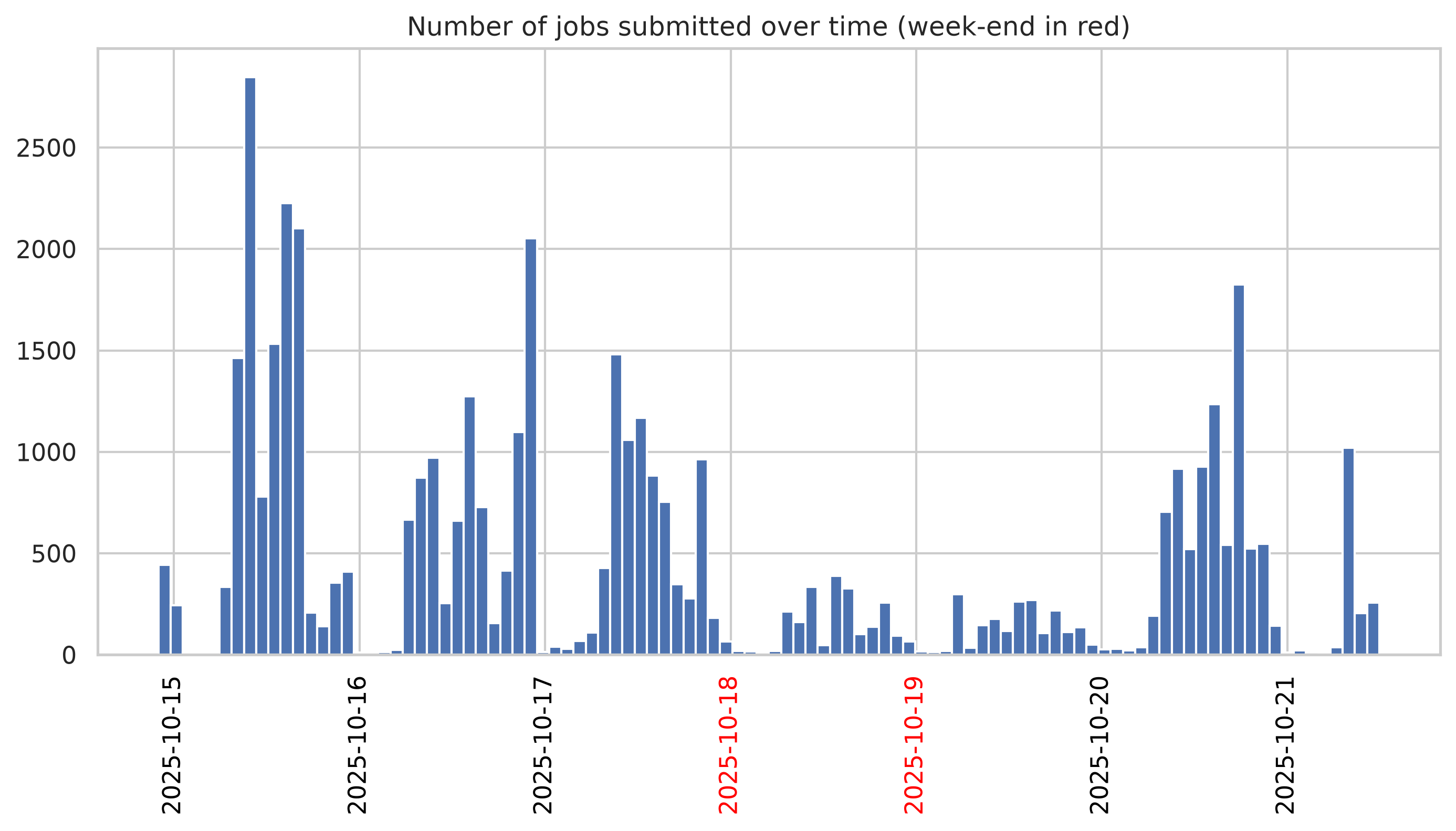}
    \caption{Number of COMPLETED jobs - Jean Zay (IDRIS)
        \label{job_over_time}}
\end{figure}

\subsection{Disparity Between Users}\label{disparity-between-users}

Regarding resource sharing and waiting times, we observed a high
disparity between users. A small number of users are responsible for the
vast majority of both the total cores$\times$seconds and the
cumulative waiting time of each partition. This is displayed in Figures \ref{total_sec_cpus_per_user}
and \ref{total_wt_per_user}, which are mass-count plots of cores$\times$seconds and waiting time per partition.
As described in \cite{ref-feitelsonWorkloadModelingComputer2015}, mass-count plots illustrate how
much a distribution is skewed, by displaying how much users are responsible for the total waiting times
(or cores$\times$seconds). A uniform distribution among users would result in the red dotted line. For waiting times, we
observe that, among all partitions, 80\% of users are only responsible for less than 10\% of all waiting times.
The distribution of cores$\times$seconds is less unbalanced, but we still observe that 80\% of users are responsible for less than 30\%
of all cores$\times$seconds, for each partition. Such unbalanced resource usage is not particularly surprising in
a supercomputer, where resources are shared between project allocations, and are not expected to be uniform.
On waiting times, the mass-count plot illustrates a highly skewed distribution, which is very common
for waiting times in HPC. Moreover, the concentration suggests that user-specific behavior is a primary determinant
of latency, and that part of the skewness could be explained by this behavior. This is confirmed later by the job burst phenomenon.

\begin{figure}
    \centering
    \includegraphics[width=\linewidth]{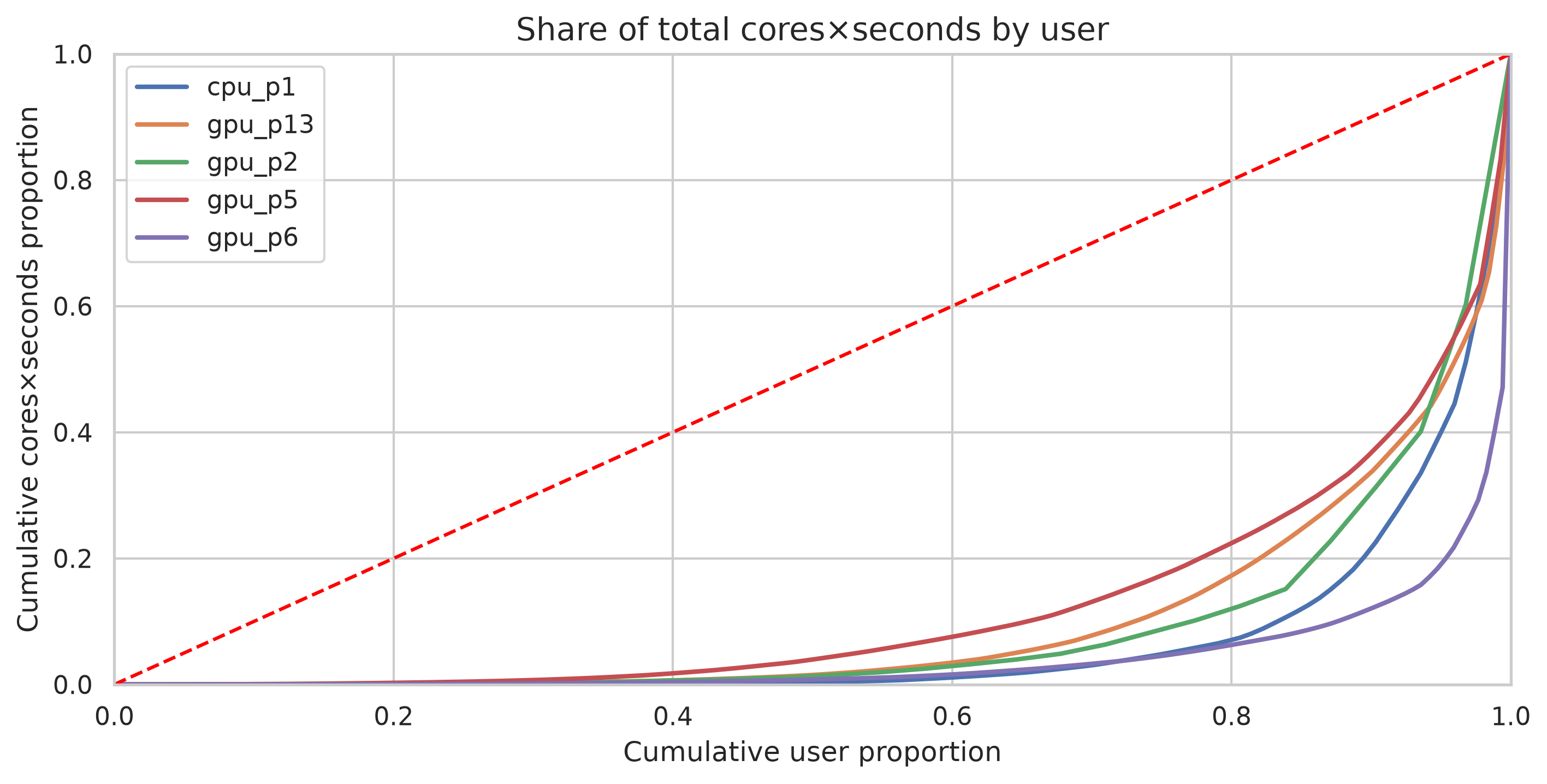}
    \caption{Mass-count plot of cores$\times$seconds - Jean Zay (IDRIS) - 16/10/25 20/10/25
        \label{total_sec_cpus_per_user}}
\end{figure}

\begin{figure}
    \centering
    \includegraphics[width=\linewidth]{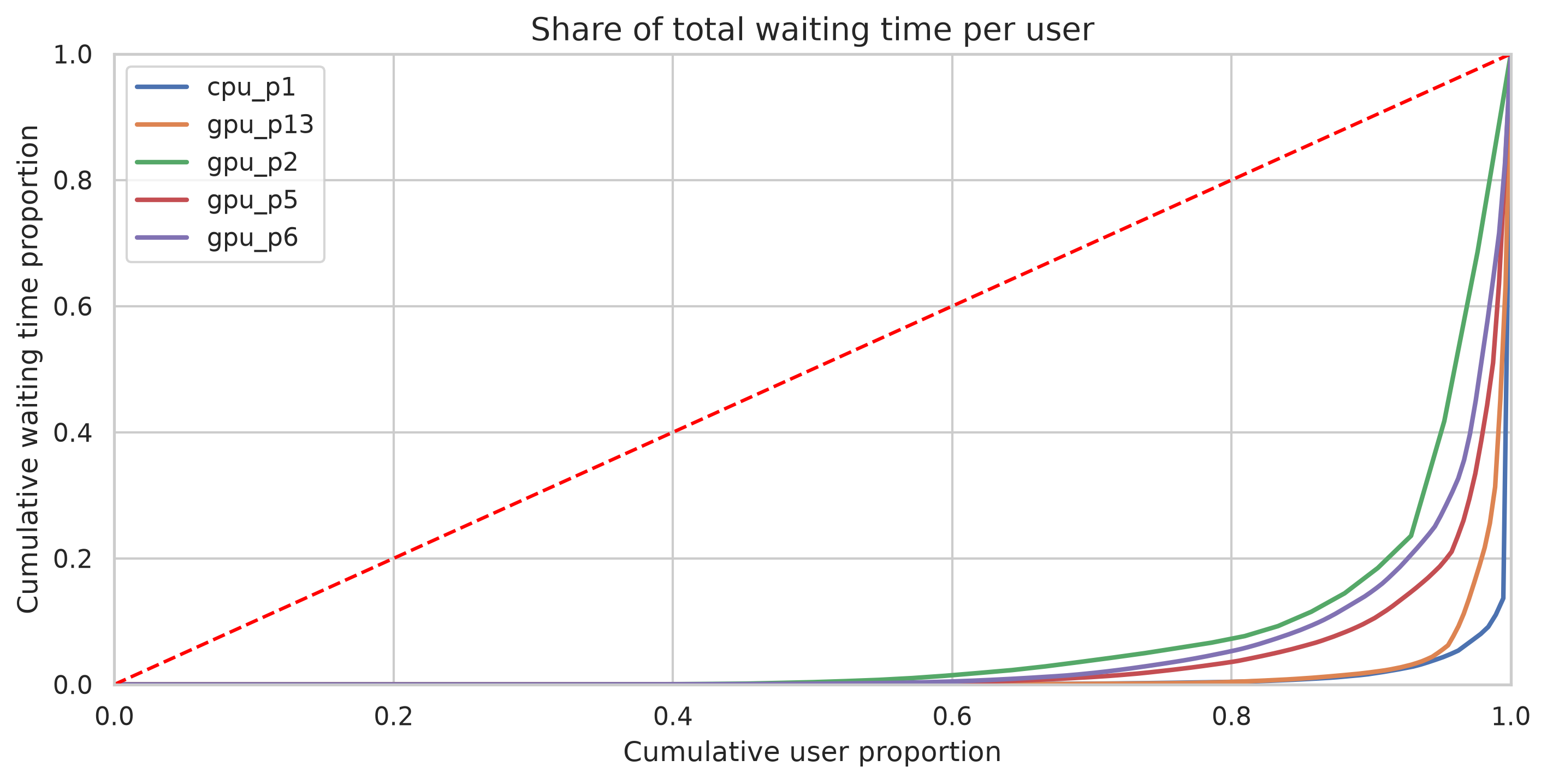}
    \caption{Mass-count plot of waiting times - Jean Zay (IDRIS) - 15/10/25 21/10/25
        \label{total_wt_per_user}}
\end{figure}

\section{The Phenomenon of Job
  Bursts}\label{the-phenomenon-of-job-bursts}

A significant discovery in the analysis of Jean Zay workload traces is
the prevalence of ``job bursts,'' defined as the submission of a series
of identical jobs at a single timestamp, by the same user. This is a common
behavior on HPC supercomputers, yet its frequency and its impact on
waiting times was, to our knowledge, never measured.
Slurm proposes a ``job array'' feature, which allows users
to submit a batch of jobs. The bursts observed here
were not part of job arrays. To reconstruct bursts, jobs with
the same number of cores and GPUs, submitted by the same user in
a 5 minutes time range, in the same partition, were defined as belonging to a burst.

The analysis reveals that these bursts are not anomalous but represent a
dominant pattern in the dataset. A major portion of all submitted jobs
is part of such bursts, with some individual sequences exceeding 100
identical jobs (Figure \ref{job_burst_perc}). Moreover, few bursts contain many
jobs, as displayed in Figure \ref{job_burst_n}. Less than 150 bursts, each containing
more than 100 jobs, represents between 50\% and 60\% of all jobs.

Furthermore, a strong linear correlation exists between a job's rank
within a burst and its observed waiting time. As the size of the burst
increases, the predictability of the waiting time via linear regression
improves, suggesting a systematic penalty for bulk submissions. This
behavior highlights a significant user-dependency in the scheduling
process: the scheduler effectively penalizes the submission of large
batches of similar jobs to prevent any single account from monopolizing
the machine. This systemic behavior is consistent with observations in
other workload traces
\cite{ref-lovellHierarchicalDeepLearning2024}. In Figure
\ref{job_burst_corr}, we display the distribution of the correlation
coefficient among all job bursts. The R² score of all linear regressions
follows the same trend as the Pearson correlation coefficient, as displayed in Figure
\ref{job_burst_r2}.

\begin{figure}
    \centering
    \includegraphics[width=\linewidth]{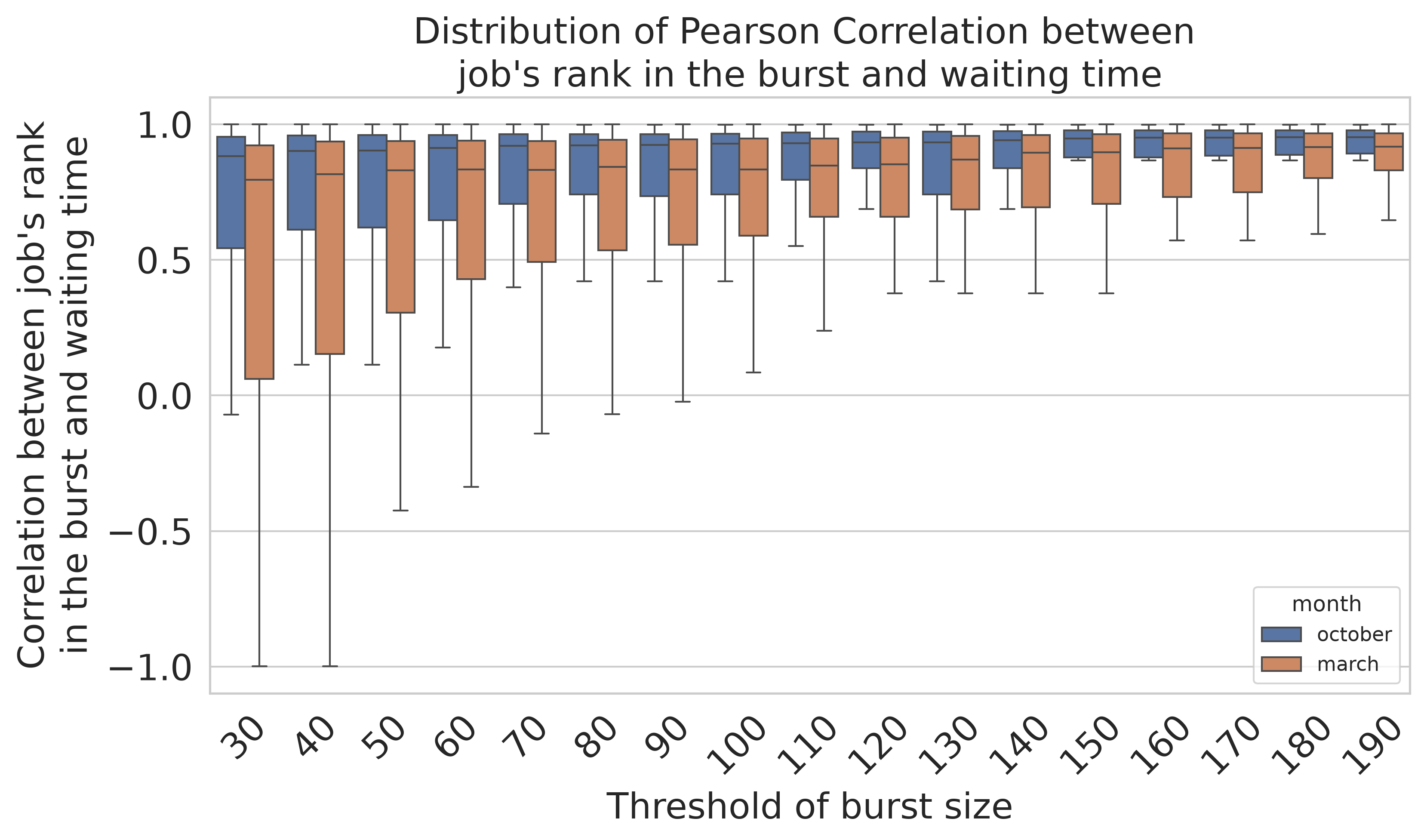}
    \caption{Correlation between waiting time and rank in the burst, among bursts
        above a given size \label{job_burst_corr}}
\end{figure}

As shown in Figure \ref{job_burst_perc}, long job bursts concern a high proportion of the
jobs. More than 40\% of the trace falls within bursts of size greater than
185. We removed jobs belonging to bursts with a correlation
coefficient above 0.8, and a length greater than 30. The remaining trace contains 30,000 jobs for October
and 45,000 for March. This respectively removed 55\% and 44\% of jobs. The dropped jobs were part
of only 100 burst for October, and 165 for March.

We conclude that, in our trace, waiting times are primarily determined by a job's rank
within its burst. This implies that waiting times are highly dependent among jobs submitted by
the same user: two jobs with identical submission parameters (number of cores, GPUs, \dots)
may exhibit vastly different waiting times. This shows
that waiting times are not i.i.d. (independent and identically distributed).
Performing a train-test split, or even computing the mean waiting time on traces,
can be misleading, since it will be strongly affected
by user-specific behavior. A sound approach, as described in \cite{ref-lovellHierarchicalDeepLearning2024},
is to carefully handle these bursts before train-test splitting.

\begin{figure}
    \centering
    \includegraphics[width=\linewidth]{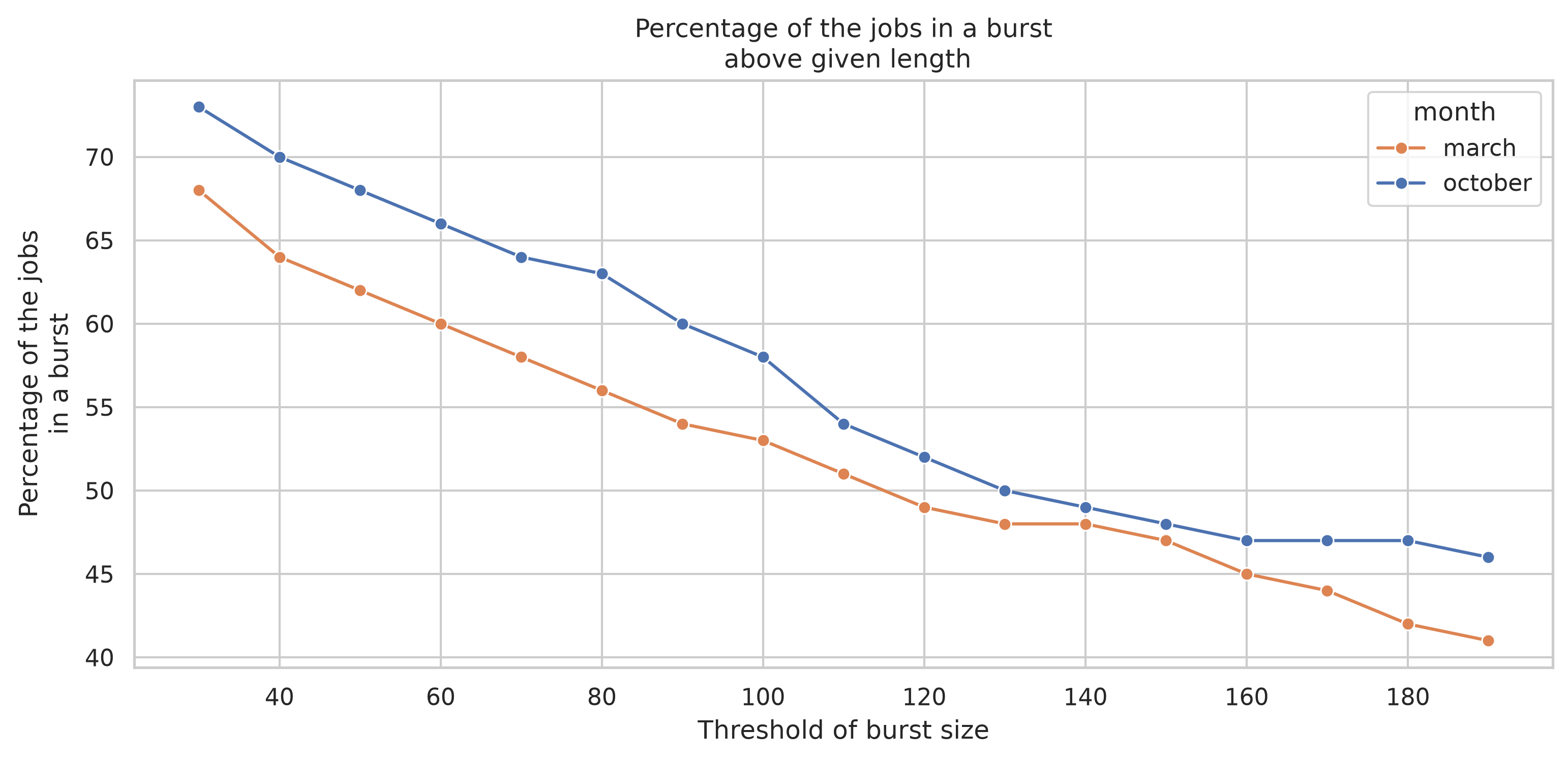}
    \caption{Percentage of the jobs belonging in bursts \label{job_burst_perc}}
\end{figure}

\begin{figure}
    \centering
    \includegraphics[width=\linewidth]{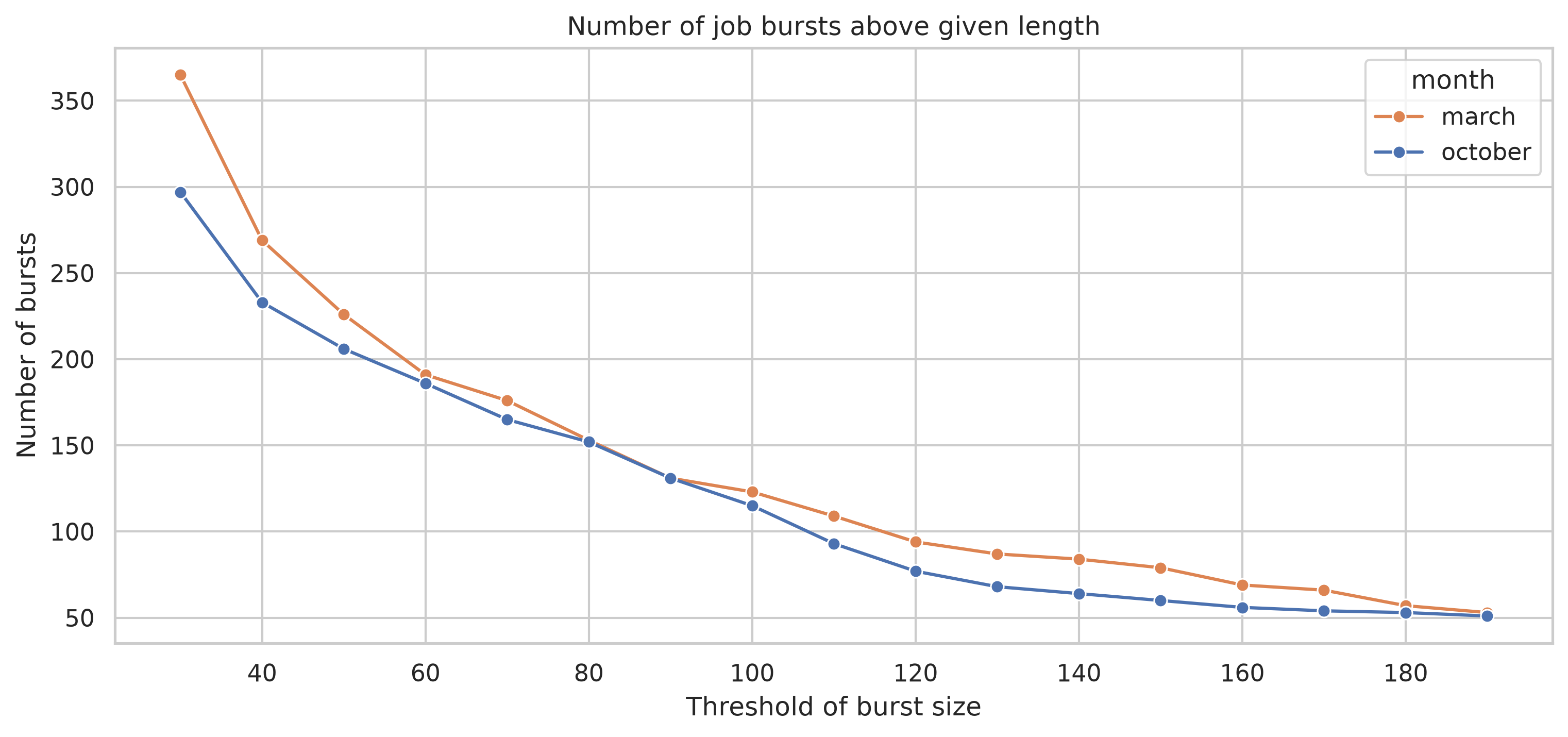}
    \caption{Number of job bursts \label{job_burst_n}}
\end{figure}

\begin{figure}
    \centering
    \includegraphics[width=\linewidth]{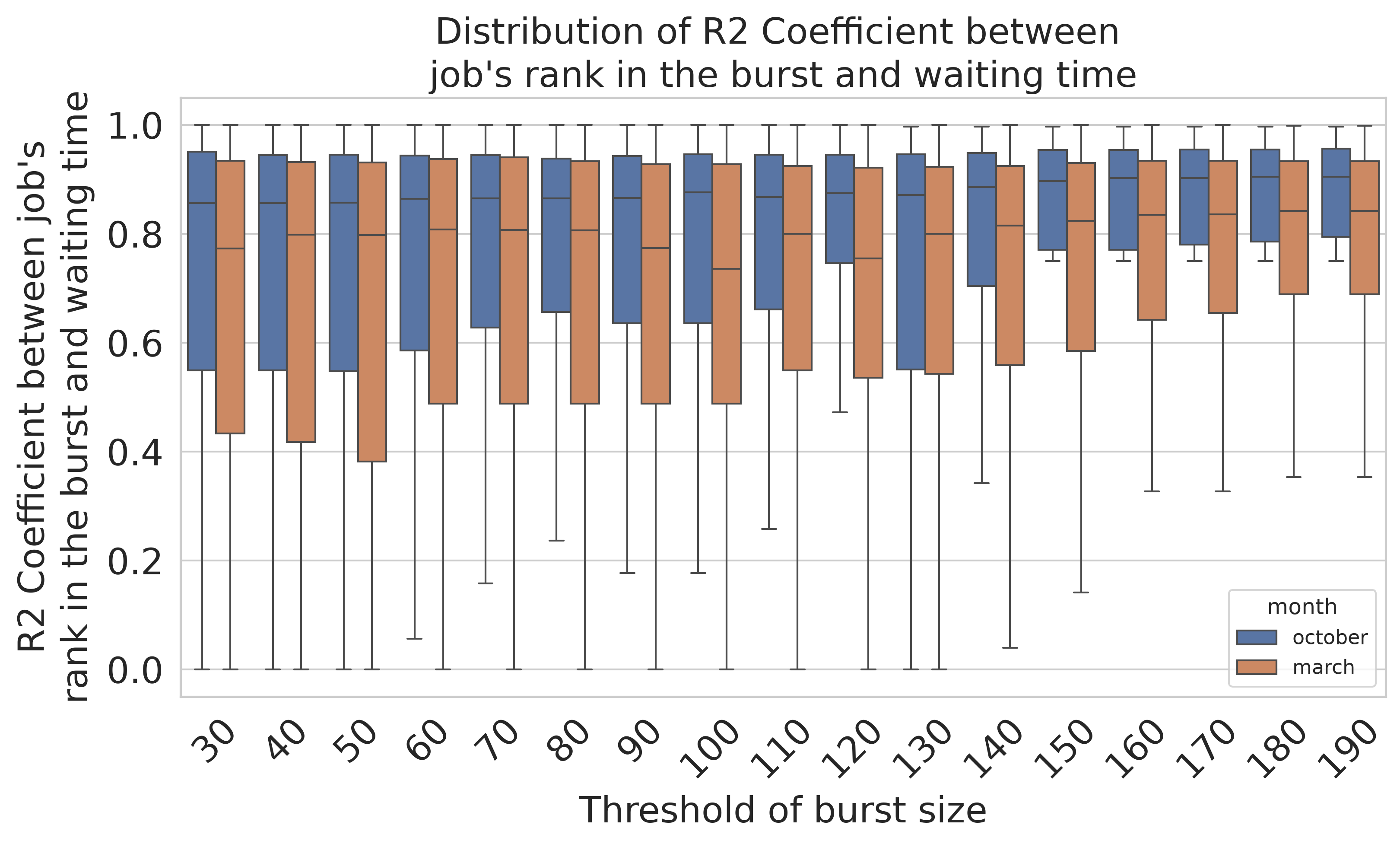}
    \caption{R2 of correlation between waiting time and rank in the burst, among bursts
        above a given size
        \label{job_burst_r2}}
\end{figure}

\section{Slurm Fairshare Factor}\label{slurm-fairshare-factor}

The analysis above suggests a strong dependency of user behavior towards waiting times.
A natural explanation could be Slurm's fairshare mechanism,
which prioritizes users based on their recent resource consumption and their attribution.

The fairshare factor is computed using what is called in Slurm's documentation the ``user historical usage.''
This historical usage is computed using a half-life formula, detailed in
Appendix \ref{appendix-a-theoretical-analysis-of-slurm-fairshare}. Proposition \ref{fs-exact}
shows that the historical usage is an exponentially decaying
function, with the half-life defined in Slurm. On Jean Zay, the
half-life is 14 days, meaning that a recently accumulated resource usage
only contributes half of its value to the fairshare computation after 14 days.

Nevertheless, the ``burst rank'' penalty identified in the previous
section operates on a much shorter timescale (hours to days). This is
illustrated in Figure \ref{historical_with_wt}, where each dot
represents a job, with the x-axis giving the job's start time (not its
creation time). We observe in this plot that several job bursts exhibit
very similar behavior, even as the user historical usage increases. If
the fairshare factor were solely responsible for this burst phenomenon,
one would expect higher waiting times for the second and third bursts,
as the historical usage would not have had sufficient time to decrease.
Instead, the waiting-time penalty appears to be ``forgotten'' between
separate bursts. Note that the job characteristics here are the same
across bursts: number of cores, partition, user, and account.

\begin{figure}
    \centering
    \includegraphics[width=\linewidth]{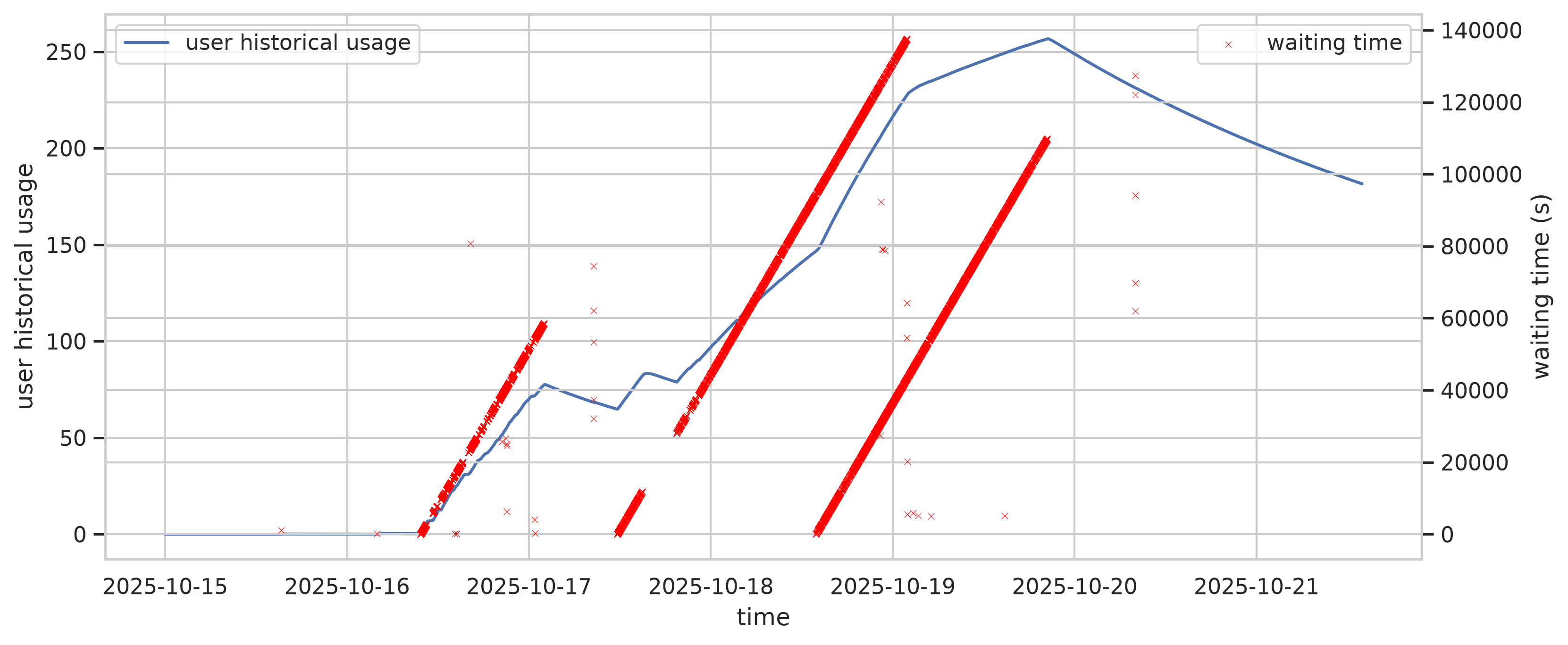}
    \caption{Example user historical usage (see Appendix \ref{appendix-a-theoretical-analysis-of-slurm-fairshare}) and waiting time of jobs \textbf{for a single user}
        \label{historical_with_wt}}
\end{figure}

\section{Analysis of Waiting Times}\label{analysis-of-waiting-times}

On traces with bursts removed (leaving 30,000 jobs for October
and 45,000 for March), we show here that jobs submitted during periods
of lower partition occupancy tend to have short waiting times. The same is observed for jobs with
fewer cores, though this effect is less consistent between the
two traces and the partitions.

\subsection{Impact of Partition Choice}\label{impact-of-partition-choice-and-qos}

We observe heterogeneous waiting times between partitions
and between the two observed periods (March and October). This is illustrated in
Figure \ref{mean_wt}, where the \texttt{gpu\_p13} partition has very low waiting time in October,
compared to other partitions. In March, however, the waiting time is significantly higher.

\begin{figure}
    \centering
    \includegraphics[width=\linewidth]{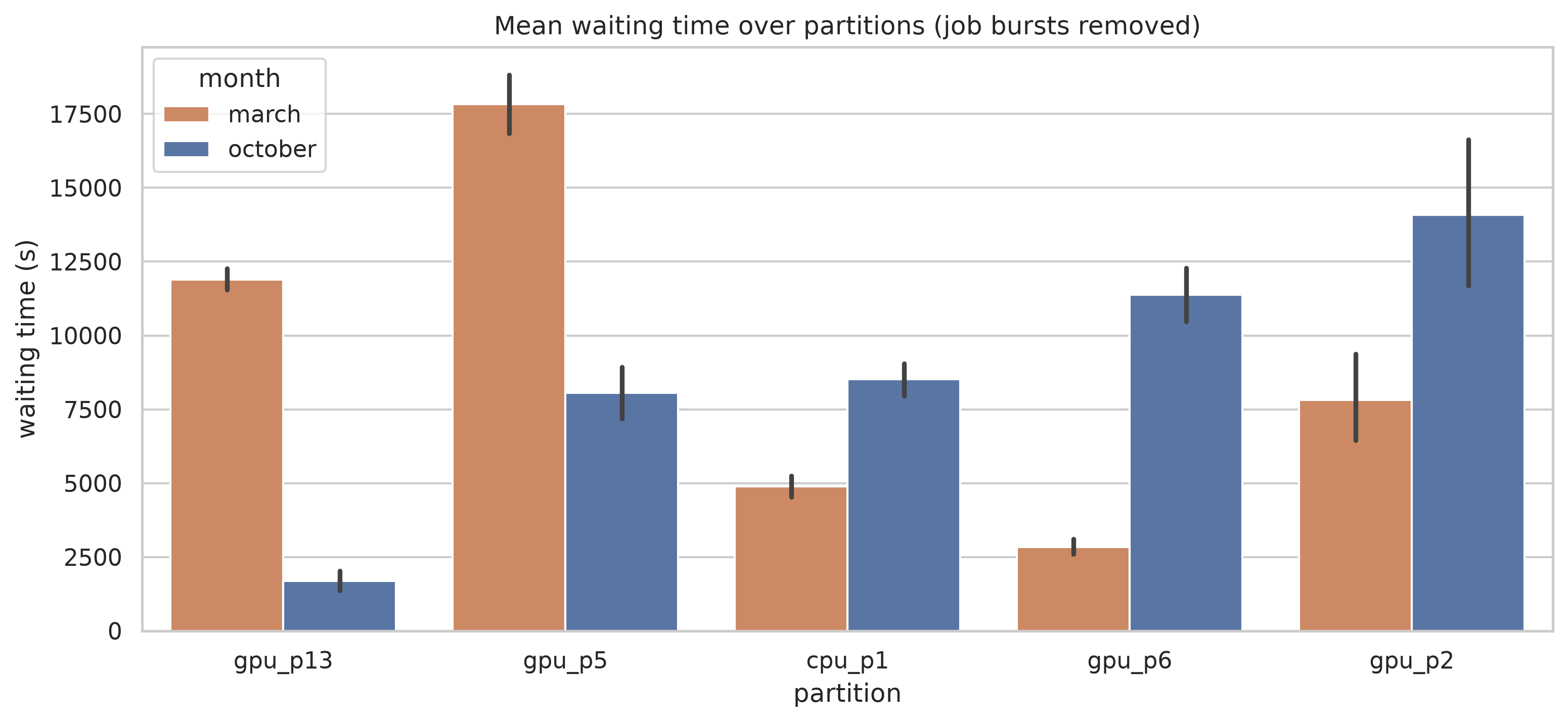}
    \caption{Mean Waiting Time per Partition}\label{mean_wt}
\end{figure}


\subsection{Impact of Partition Occupancy}\label{impact-of-partition-usage}

The low average waiting time on \texttt{gpu\_p13} in October could be related to the
reduced weekend occupancy (Figure \ref{jz_charge}). To investigate this, we splitted the jobs in two groups.
Group \texttt{H} (High) with waiting times above 45 seconds, and group \texttt{L} (Low), with waiting times below 45 seconds, following a
strategy similar to \cite{ref-lovellHierarchicalDeepLearning2024}, \cite{ref-brownPredictingBatchQueue2022}, \cite{kumar}.
The threshold was chosen to balance the two classes, \textbf{after removing job bursts}.
As shown in Figure \ref{hist_log_wt}, removing job bursts largely eliminates
long waiting times. This results in a significantly smaller threshold than
\cite{ref-lovellHierarchicalDeepLearning2024} (10 minutes).

\begin{figure}
    \centering
    \includegraphics[width=\linewidth]{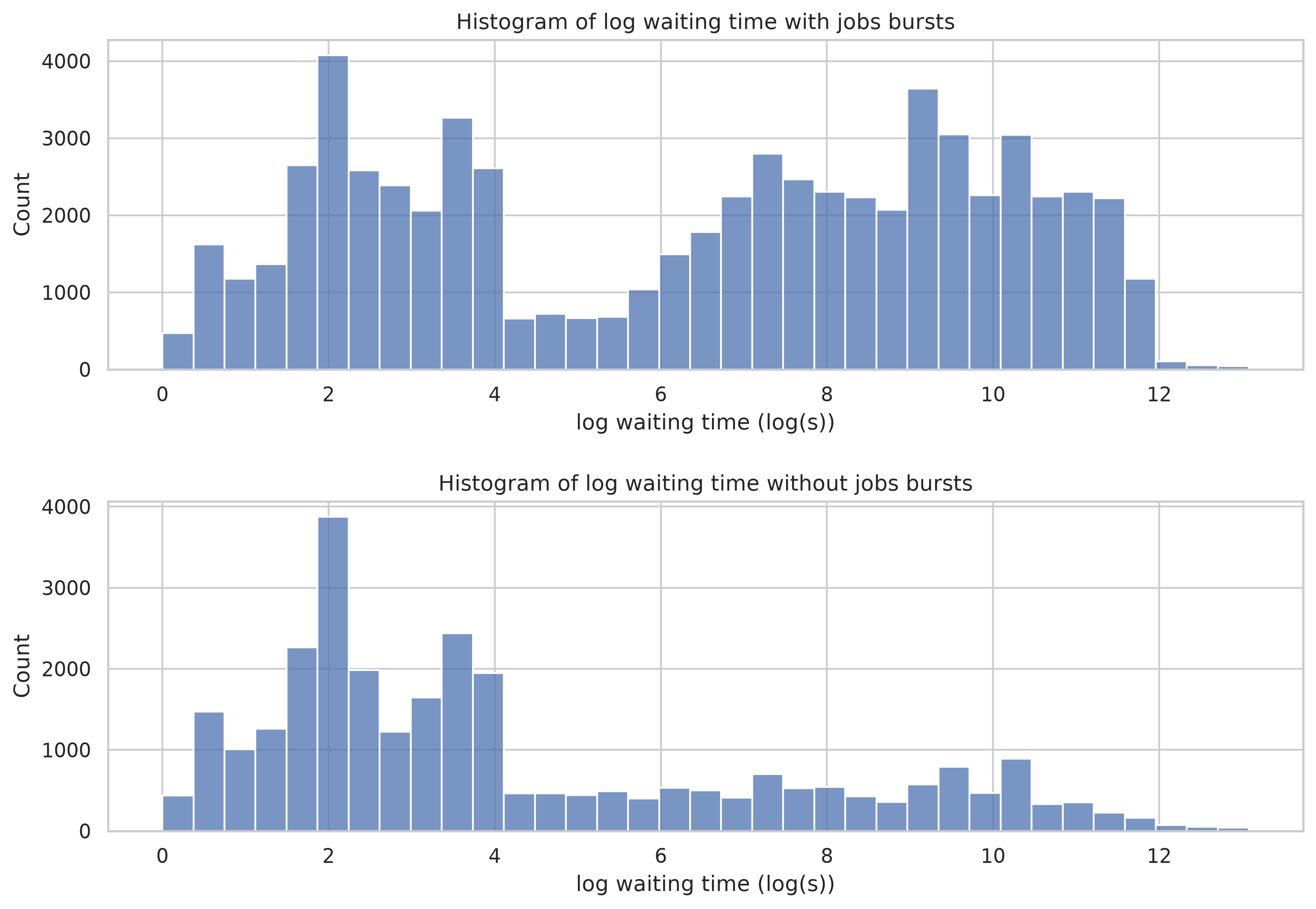}
    \caption{Histogram of log waiting time, before and after job burst
        removal - 15/10/26 21/10/26\label{hist_log_wt}}
\end{figure}

\begin{figure}
    \centering
    \includegraphics[width=\linewidth]{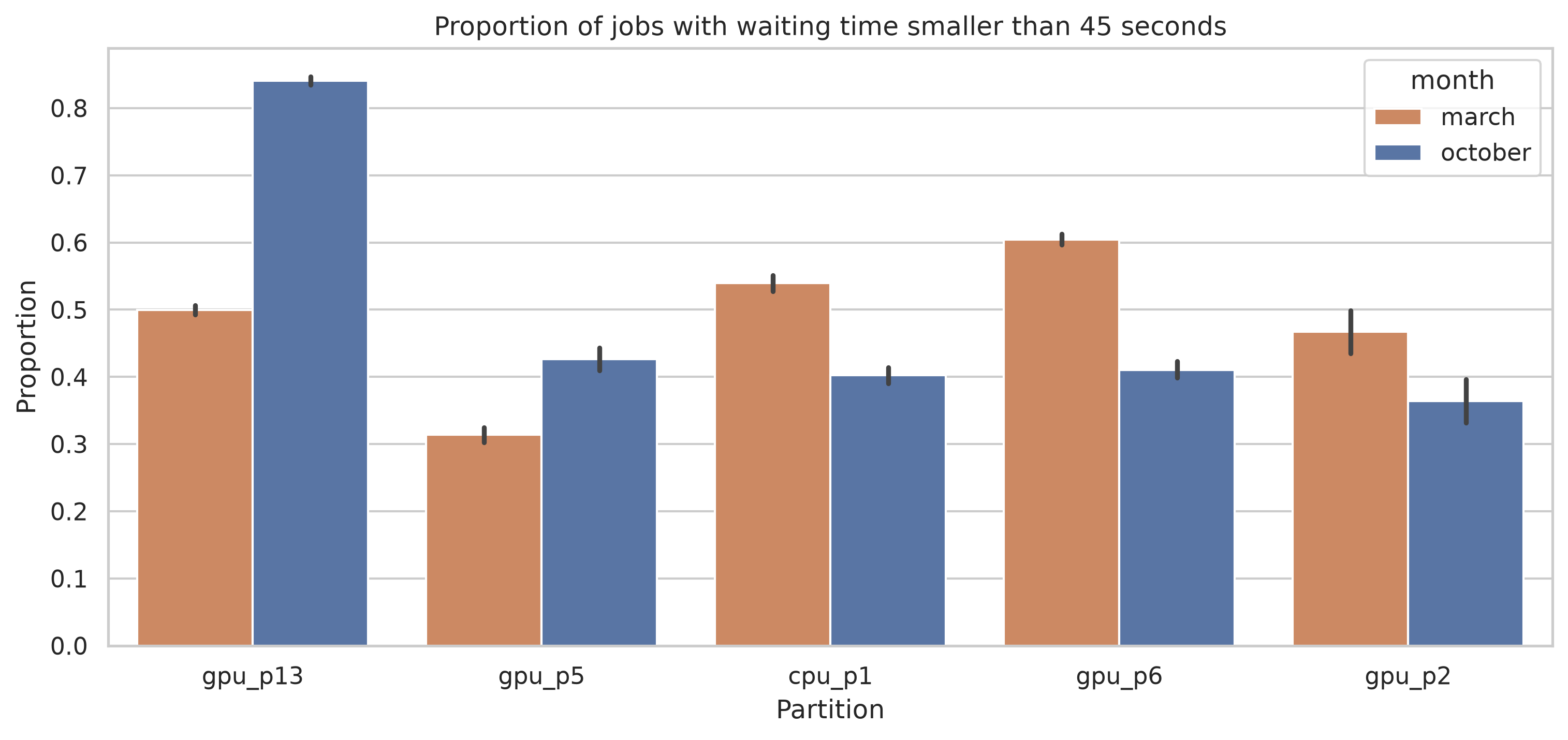}
    \caption{Proportion of jobs with waiting times smaller than 45 seconds,
        per partition \label{small_wt_prop}}
\end{figure}

In Tables \ref{tab:ttest-oct} and \ref{tab:ttest-march}, we show the results of Welch's
t-test\footnote{https://en.wikipedia.org/wiki/Welch's\_t-test} between the two groups.
The null hypothesis (H\textsubscript{0}) is that $\mu_{\texttt{H}} < \mu_{\texttt{L}}$, i.e., that the mean occupancy rate
of group \texttt{H} is lower than that of group \texttt{L}.
The Tables \ref{tab:ttest-oct} and \ref{tab:ttest-march} contain respectively the
size of group \texttt{H} and group \texttt{L} ($n_{\texttt{H}}$, $n_{\texttt{L}}$), their means and standard deviation ($\mu_{\texttt{H}}$, $\mu_{\texttt{L}}$, $\sigma_{\texttt{H}}$, $\sigma_{\texttt{L}}$).
the $t$ statistic and the $p$-value of the test (computed with scikit learn
\footnote{https://docs.scipy.org/doc/scipy/reference/generated/scipy.stats.ttest\_ind.html}).
A $p$-value below $0.05$ is interpreted as a rejected test, meaning that, significantly,
$\mu_{\texttt{H}} \geq \mu_{\texttt{L}}$.
In October, on partition \texttt{gpu\_p13}, jobs with short waiting times were submitted
on average on periods with 61\% occupancy, while the long ones were submitted on period
with average occupancy of 89\% occupancy. A similar drop is observed for \texttt{gpu\_p5},
which is the other partition with less occupancy during weekend (Figure \ref{jz_charge}).
For the March trace, the partitions with a lower occupancy between the 2 and 3 of March are those
for which the test is rejected. These tests tends to confirm that when a
partition has a very low occupancy, waiting times are smaller. Future works
could try to quantify this trend.

\begin{figure}
    \centering
    \includegraphics[width=\linewidth]{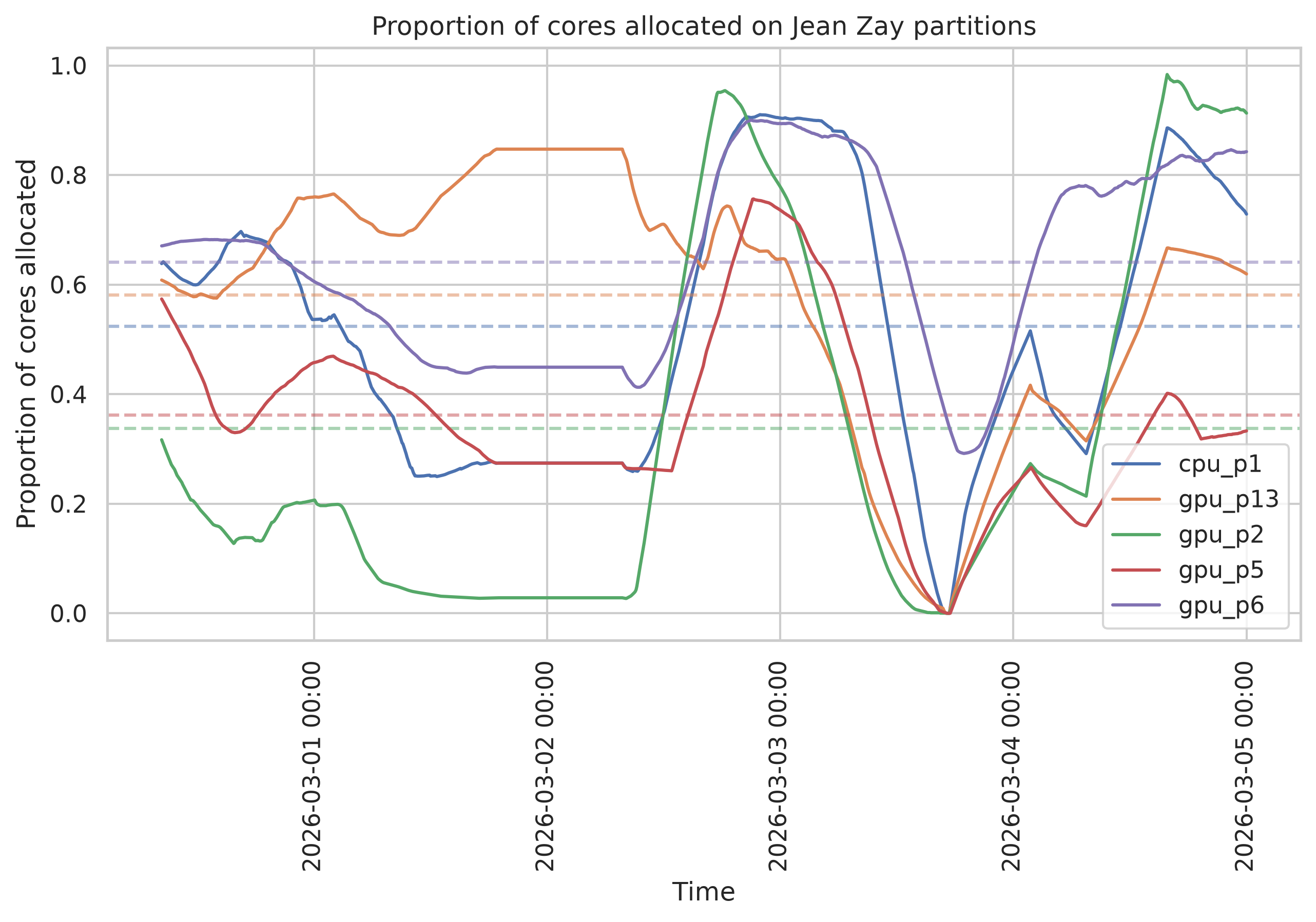}
    \caption{Occupancy rate of each partition - Jean Zay (IDRIS). A reboot of the supercomputer happened the 03/03/2026. \label{jz_charge_march}}
\end{figure}

\begin{table}[ht]
    \centering
    \caption{Mean partition usage among jobs with waiting time below 45 seconds ($\mu_{\texttt{L}}$) and others ($\mu_{\texttt{H}}$), October trace}
    \label{tab:ttest-oct}
    \begin{tabularx}{\columnwidth}{lrrrrrrrr}
        \toprule
        Partition & $n_{\texttt{H}}$ & $n_{\texttt{L}}$ & $\mathbf{\mu_{\texttt{H}}}$ & $\mathbf{\mu_{\texttt{L}}}$ & $\sigma_{\texttt{H}}$ & $\sigma_{\texttt{L}}$ & $t$  & $\mathbf{p}$    \\
        \midrule
        cpu\_p1   & 1975             & 1097             & \textbf{0.77}               & \textbf{0.78}               & 0.1                   & 0.1                   & -2.4 & \textbf{0.9917} \\
        gpu\_p13  & 858              & 5662             & \textbf{0.89}               & \textbf{0.61}               & 0.1                   & 0.2                   & 60.5 & \textbf{0.0000} \\
        gpu\_p2   & 276              & 103              & \textbf{0.78}               & \textbf{0.81}               & 0.1                   & 0.1                   & -2.4 & \textbf{0.9901} \\
        gpu\_p5   & 920              & 582              & \textbf{0.89}               & \textbf{0.73}               & 0.1                   & 0.2                   & 17.9 & \textbf{0.0000} \\
        gpu\_p6   & 1791             & 1053             & \textbf{0.92}               & \textbf{0.91}               & 0.1                   & 0.1                   & 4.6  & \textbf{0.0000} \\
        \bottomrule
    \end{tabularx}
\end{table}
\begin{table}[ht]
    \centering
    \caption{Mean partition usage among jobs with waiting time below 45 seconds ($\mu_{\texttt{L}}$) and others ($\mu_{\texttt{H}}$), March trace}
    \label{tab:ttest-march}
    \begin{tabularx}{\columnwidth}{lrrrrrrrr}
        \toprule
        Partition & $n_{\texttt{H}}$ & $n_{\texttt{L}}$ & $\mathbf{\mu_{\texttt{H}}}$ & $\mathbf{\mu_{\texttt{L}}}$ & $\sigma_{\texttt{H}}$ & $\sigma_{\texttt{L}}$ & $t$  & $\mathbf{p}$    \\
        \midrule
        cpu\_p1   & 1565             & 3167             & \textbf{0.83}               & \textbf{0.73}               & 0.2                   & 0.2                   & 18.4 & \textbf{0.0000} \\
        gpu\_p13  & 6796             & 7653             & \textbf{0.65}               & \textbf{0.65}               & 0.2                   & 0.1                   & -0.1 & \textbf{0.5244} \\
        gpu\_p2   & 394              & 294              & \textbf{0.72}               & \textbf{0.54}               & 0.4                   & 0.4                   & 6.4  & \textbf{0.0000} \\
        gpu\_p5   & 3348             & 1589             & \textbf{0.48}               & \textbf{0.48}               & 0.2                   & 0.2                   & 1.0  & \textbf{0.1545} \\
        gpu\_p6   & 2850             & 5324             & \textbf{0.77}               & \textbf{0.70}               & 0.1                   & 0.2                   & 20.3 & \textbf{0.0000} \\
        \bottomrule
    \end{tabularx}
\end{table}

\subsection{Impact of resources}

We use the same methodology as above, replacing the occupancy
by the number of cores allocated to each job. We show that, in the October trace, for the \texttt{cpu\_p1} partition,
the average number of cores in jobs with waiting times smaller than 45 seconds is 35, versus
102 for jobs with longer waiting times. Results are summarized in Table \ref{tab:ttest-oct-2} and \ref{tab:ttest-march-2}.

\begin{table}[ht]
    \centering
    \caption{Mean number of core among jobs with waiting time below 45 seconds ($\mu_{\texttt{L}}$) and others ($\mu_{\texttt{H}}$), October trace}
    \label{tab:ttest-oct-2}
    \begin{tabularx}{\columnwidth}{lrrrrrrrr}
        \toprule
        Partition & $n_{\texttt{H}}$ & $n_{\texttt{L}}$ & $\mathbf{\mu_{\texttt{H}}}$ & $\mathbf{\mu_{\texttt{L}}}$ & $\sigma_{\texttt{H}}$ & $\sigma_{\texttt{L}}$ & $t$  & $\mathbf{p}$    \\
        \midrule
        cpu\_p1   & 3954             & 2665             & \textbf{102}                & \textbf{35}                 & 330                   & 113                   & 11.8 & \textbf{0.0000} \\
        gpu\_p13  & 2112             & 11168            & \textbf{26}                 & \textbf{18}                 & 100                   & 49                    & 3.4  & \textbf{0.0004} \\
        gpu\_p2   & 470              & 269              & \textbf{10}                 & \textbf{5}                  & 13                    & 4                     & 8.7  & \textbf{0.0000} \\
        gpu\_p5   & 1876             & 1395             & \textbf{24}                 & \textbf{13}                 & 60                    & 15                    & 7.0  & \textbf{0.0000} \\
        gpu\_p6   & 3446             & 2399             & \textbf{149}                & \textbf{48}                 & 725                   & 62                    & 8.2  & \textbf{0.0000} \\
        \bottomrule
    \end{tabularx}
\end{table}
\begin{table}[ht]
    \centering
    \caption{Mean number of core among jobs with waiting time below 45 seconds ($\mu_{\texttt{L}}$) and others ($\mu_{\texttt{H}}$), March trace}
    \label{tab:ttest-march-2}
    \begin{tabularx}{\columnwidth}{lrrrrrrrr}
        \toprule
        Partition & $n_{\texttt{H}}$ & $n_{\texttt{L}}$ & $\mathbf{\mu_{\texttt{H}}}$ & $\mathbf{\mu_{\texttt{L}}}$ & $\sigma_{\texttt{H}}$ & $\sigma_{\texttt{L}}$ & $t$  & $\mathbf{p}$    \\
        \midrule
        cpu\_p1   & 3167             & 3709             & \textbf{121}                & \textbf{58}                 & 458                   & 427                   & 5.9  & \textbf{0.0000} \\
        gpu\_p13  & 9152             & 9133             & \textbf{14}                 & \textbf{13}                 & 44                    & 30                    & 1.4  & \textbf{0.0858} \\
        gpu\_p2   & 528              & 463              & \textbf{8}                  & \textbf{8}                  & 7                     & 8                     & -0.3 & \textbf{0.6190} \\
        gpu\_p5   & 4156             & 1901             & \textbf{11}                 & \textbf{9}                  & 15                    & 4                     & 8.8  & \textbf{0.0000} \\
        gpu\_p6   & 5175             & 7901             & \textbf{66}                 & \textbf{46}                 & 198                   & 80                    & 6.9  & \textbf{0.0000} \\
        \bottomrule
    \end{tabularx}
\end{table}

\section{Conclusion}\label{conclusion}

Understanding waiting times in HPC is critical for cross-facility
scientific workflows deployed in highly variable environments. We
presented a statistical analysis of waiting times on Jean Zay, one
of France's three Tier-0 supercomputers.

We found that job bursts account for the majority of long waiting times, with
more than 40\% of them being well described by simple linear regressions. This
highlights the importance of accounting for this phenomenon before applying
statistical or machine-learning methods. We also showed that waiting times tend
to be shorter during periods of low system occupancy, and for jobs requesting
small number of cores. Finally, we observed strong
heterogeneity across the two weeks under study as well as between partitions,
raising concerns about the generalizability of models trained on limited datasets.

These findings underline an important challenge for future cross-facility workflow
schedulers: access to clean workload trace datasets, ideally in near real time.
This is particularly challenging because Tier-0
supercomputers operate in highly restricted environments, and workload traces often
contain highly identifying information that cannot always be released while complying
with the terms of use (CGU).

Future work could extend this study to longer and other workload traces. In particular, the impact
of job bursts on waiting times could be investigated on other traces. It would also
be valuable to quantify the relationship between partition occupancy and waiting times.
Furthermore, the influence of user-specific factors, such as fairshare,
could be evaluated using workload traces spanning at least one month, given that the
half-life of user historical usage is 14 days.

\section{Appendix A: Theoretical Analysis of Slurm
  Fairshare}\label{appendix-a-theoretical-analysis-of-slurm-fairshare}

In this appendix, we provide the theoretical derivation of the Slurm
fairshare mechanism to justify why it is unlikely to be the primary
cause of the job burst penalty. Namely, we prove that user historical usage can be
seen as the convolution product between user resource usage and an exponential
decay. This allows to compute an exact formula, and to prove that it is
a piecewise exponential function.

\subsection{Slurm Priority and
    Fairshare}\label{slurm-priority-and-fairshare}

Slurm computes the priority of a job based on several factors: age in
the queue, job size, fairshare, and others \footnote{\url{https://slurm.schedmd.com/priority\_multifactor.html}}.
These are combined into a weighted sum. We focus here on the ``classic
fairshare algorithm'' \footnote{\url{https://slurm.schedmd.com/classic\_fair\_share.html}}
employed on Jean Zay.

The fairshare factor is a real number between 0 and 1: values below 0.5
indicate over-consumption, while values above 0.5 indicate
under-consumption.

\subsection{User Historical Usage}\label{user-historical-usage}

The main ingredient in the computation of the fairshare factor is the
user historical usage, computed as a rolling weighted mean of the
cores$\times$seconds previously allocated to each user. Let \(N_u\) be the step
function representing the number of cores allocated to user
\(u \in \mathcal U\) over time.

\begin{definition}
    For each user $u \in \mathcal U$, and each time stamp $t \in \mathbb R$, the user historical usage at time $t$ is:
    \begin{align*}
        UH_u(t,\delta)
         & = \sum_{k \in \mathbb{N}} RU(t,\delta,k)D^{k\delta} \\
         & = RU(t,\delta,0)
        + RU(t,\delta,1)D^{\delta}                             \\
         & \quad + RU(t,\delta,2)D^{2\delta}
        + \cdots
    \end{align*}
    Where $RU(t, \delta, k)$ is the raw utilization (cores$\times$seconds) over the $k$-th time period of length $\delta$ before $t$, and $D$ is the decay factor. On Jean Zay, the half-life $\alpha$ is 14 days, $D = \exp(-\ln(2)/\alpha)$ and $\delta$ is 5 minutes.
\end{definition}

We derive an alternative definition based on a convolution product:

\begin{proposition}
    Let $g(t) = \mathbf 1_{\mathbb R^+}(t) \exp(-\lambda t)$, with $\lambda = \ln(2)/\alpha$. As $\delta \to 0$, the historical usage converges in $L^{\infty}$ to:
    $$
        UH_u(t) := (N_u*g)(t) = \int_{\mathbb R} N_u(s)g(t-s)ds
    $$
    \label{prop:fs_alt_def}
\end{proposition}

\begin{IEEEproof}
    From the above definition, we write, for $k \in \mathbb N$ and $t \in \mathbb R$, $RU(t, \delta, k)$ as :
    $$
        RU(t, \delta, k) = \int_{\mathbb R} \mathbf 1_{]t-(k+1)\delta, t-k\delta]}(s)N(s)ds.
    $$

    Hence, $$
        UH_u(t, \delta) = \sum_{k \in \mathbb N} \int_{\mathbb R} \mathbf 1_{]t-(k+1)\delta, t-k\delta]}(s) N(s) D^{k\delta}ds.
    $$

    Since all integrands $\mathbf 1_{]t-(k+1)\delta, t-k\delta]}(s)N(s) D^{k\delta}$ are positive, the Beppo Levi theorem allows us to exchange the sum and the integral, and gives :

    $$
        UH_u(t, \delta) = \int_{\mathbb R} \sum_{k \in \mathbb N} \mathbf 1_{]t-(k+1)\delta, t-k\delta]}(s)N(s) D^{k\delta}ds
    $$
    $$
        UH_u(t, \delta) = \int_{\mathbb R} N(s) (\sum_{k \in \mathbb N} \mathbf 1_{]t-(k+1)\delta, t-k\delta]}(s) D^{k\delta})ds
    $$
    $$
        UH_u(t, \delta) = \int_{\mathbb R} N(s) D^{\lceil \frac{t-s}{\delta}\rceil \delta} \mathbf 1_{\mathbb R_+}(t-s)ds
    $$

    For a given $s \in \mathbb R$, and $k \in \mathbb N$, $\mathbf 1_{]t-(k+1)\delta, t-k\delta]}(s)$ is non-zero iff:

    \begin{equation}
        t-(k+1)\delta<s\leq t-k\delta,
        \label{eq:proof}
    \end{equation}

    which is equivalent to $\frac{t-s}{\delta} - 1 <k\leq \frac{t-s}{\delta}$. This happens for $s \in \mathbb R$, and $k \in \mathbb N$ iff $k=\lceil \frac{t-s}{\delta} \rceil$, and $t-s \geq 0$.

    To conclude, we extract from equation \ref{eq:proof} the fact that, $D^{(t-s) + \frac{1}{\delta}} < D^{\lceil \frac{t-s}{\delta} \rceil \delta} \leq D^{t-s}$.
    Which happens because $u \mapsto \exp(-\frac{\lambda}{\alpha}\times u)$ is decreasing on $\mathbb R$.
    Hence, $| D^{\lceil \frac{t-s}{\delta}\rceil \delta} - D^{t-s}| = D^{t-s} - D^{\lceil \frac{t-s}{\delta}\rceil \delta} \leq D^{t-s}(1- D^{\frac{1}{\delta}})$.
    So, for each $s \in \mathbb R$ and $\delta >0$, $\sup_{t \in \mathbb R}|\mathbf 1_{\mathbb R_+}(t-s) D^{\lceil \frac{t-s}{\delta}\rceil \delta} - \mathbf 1_{\mathbb R_+}(t-s)D^{t-s} |\leq 1 - D^{\frac{1}{\delta}}$.
    The sequence of functions $g_{\delta}(u)=D^{\lceil \frac{u}{\delta}\rceil \delta} \mathbf 1_{\mathbb R_+}(u)$ converges in $L^{\infty}$ to $g$. Since $UH_u(., \delta) = N_u*g_{\delta}$, we have :

    $$
        \|UH_u(., \delta) - UH_u\|_{\infty} = \|N_u*(g_{\delta}-g)\|_{\infty} \leq \|N_u\|_1\|g_{\delta}-g\|_{\infty}
    $$

    Since $N_u$ is $L^1$ (the total cores$\times$seconds is finite), it follows from the above inequality that $UH_u(., \delta)$ converges in $L^{\infty}$ towards $UH_u$.

\end{IEEEproof}

From the convolution product properties, it follows that \(UH_u\) is a continuous function (as observed
in Figure \ref{historical_with_wt}).
Moreover, the $L^1$ norm of $N_u*g$ is equal to : $\|N_u\|_1\times\|g\|_1 = \|N_u\|_1\frac{1}{\lambda}$, since both functions are positive, using the Fubini theorem.
We conclude that the $L^1$ norm of $UH_u$ is the total cores$\times$seconds
allocated to the user, divided by $\lambda$. It does not depend on any scheduling choice.

Finally, since \(N_u\) is a step function, the convolution product
\(N_u*g\) can be computed explicitly:

\begin{proposition}\label{fs-exact}
    If we write $N_u = \sum_{i = 0}^M n_i \mathbf 1_{[a_i, a_{i+1}[}$, then for $t \in [a_i, a_{i+1}[$, the historical usage is:
    $$
        UH_u(t) = \frac{1}{\lambda} n_i + \frac{1}{\lambda} B_i e^{-\lambda t}
    $$
    where $B_i = \sum_{j=0}^i e^{\lambda a_j} (n_j-n_{j-1})$ with $n_{-1} = 0$.
    \label{prop:fs_exact}
\end{proposition}

This proves that \(UH_u\) is a \textbf{piecewise exponential function}.
This slow decay (half-life of 14 days) seems to be incompatible with the rapid
``reset'' of the penalty observed between separate job bursts.

\subsection{User Normalized Usage and Fairshare
    Factor}\label{user-normalized-usage-and-fairshare-factor}

The normalized usage \(UN_u\) is defined as: \[
    UN_u(t) = \frac{UH_u(t)}{\sum_{y \in \mathcal U}UH_y(t)}
\] The final fairshare factor \(FS_u\) is then computed as
\(FS_u(t) = 2^{-\frac{UN_u(t)}{D.S_u}}\), where \(S_u\) is the user's
share and \(D\) is a dampening factor.


\end{document}